\documentclass[a4paper,12pt]{article}
\usepackage{amsmath, amsthm, amsfonts, amssymb, dsfont}
\usepackage{graphicx,multicol}
\theoremstyle{plain}
\newtheorem{theorem}{Theorem}[section]
\newtheorem{lemma}{Lemma}[section]

\theoremstyle{remark}

\def\tht{\theta}

\def\g{\gamma}
\def\G{\Gamma}
\def\l{\lambda}
\def\p{\partial}

\def\E{\mbox{\rm e}}

\def\b{\beta}

\def\L{\Lambda}

\def\Odr{\mathcal{O}}
\def\H{W_2}
\def\Hloc{W_{2,\mathrm{loc}}}

\def\di{\,\mathrm{d}}

\def\I{\mathrm{I}}
\def\iu{\mathrm{i}}

\DeclareMathOperator{\RE}{Re} \DeclareMathOperator{\IM}{Im}

\DeclareMathOperator{\supp}{supp}

\numberwithin{equation}{section}

\begin{document}

\allowdisplaybreaks


\begin{center}
\textbf{\Large Tunneling resonances in systems without \\ a classical trapping}

\bigskip

{\large\textbf{D.~Borisov$^a$, P.~Exner$^b$,
A.~Golovina$^c$}}

\end{center}

\begin{itemize}
\item[a)] \emph{Faculty of Physics and Mathematics, Bashkir
State Pedagogical University, October rev. st.~3a, Ufa, 450000, and Institute of Mathematics of Ufa Scientific Center of RAS, Chernyshevskogo str. 112, 450008, Ufa, Russia}\\
  \emph{E-mail:}\ \texttt{borisovdi@yandex.ru},
\emph{URL:} \texttt{http://borisovdi.narod.ru/}

 \item[b)] \emph{Nuclear Physics Institute, Academy of Sciences, 25068 \v{R}e\v{z} near Prague, and Doppler Institute, Czech Technical University, B\v{r}ehov\'{a}~7, 11519~Prague, Czech Republic}\\
      \emph{E-mail:}\ \texttt{exner@ujf.cas.cz}, \emph{URL:}\  \texttt{http://gemma.ujf.cas.cz/\~{}exner/}

 \item[c)] \emph{Faculty of Physics and Mathematics, Bashkir
State Pedagogical University, October rev. st.~3a, Ufa, 450000, Russia}\\
 \emph{E-mail:}\ \texttt{nastya\_gm@mail.ru}
\end{itemize}

\medskip

\begin{quote}

\textbf{Abstract:} In this paper we analyze a free quantum particle in a straight Dirichlet waveguide which has at its axis two Dirichlet barriers of lengths $\ell_\pm$ separated by a window of length $2a$. It is known that if the barriers are semiinfinite, i.e. we have two adjacent waveguides coupled laterally through the boundary window, the system has for any $a>0$ a finite number of eigenvalues below the essential spectrum threshold. Here we demonstrate that for large but finite $\ell_\pm$ the system has resonances which converge to the said eigenvalues as $\ell_\pm\to\infty$, and derive the leading term in the corresponding asymptotic expansion.

\bigskip


\end{quote}

\section{Introduction}

Tunneling resonances belong to the number of most traditional and important topics in quantum mechanics. The fact that a quantum particle can leave a region in which it is classically confined by tunneling through a potential barrier was noted in early days of the theory \cite{G} and the discovery of artificial radioactivity several years later demonstrated that it can give rise to a resonance process. Over the years, the tunneling resonance effect became a subject of countless studies, among which one can find many examples of a rigorous analysis --- see, for instance, \cite{H, AH, CDS, CDKS, E}.

Potential barrier tunneling is, however, only one of many instances where quan\-tum mechanics confronts us with effects which defy out intuition based on an everyday macroscopic experience. Another one is represented by purely geometric binding in situations where there is no classical trapping: examples are bound states in bent \cite{ES} or laterally coupled \cite{ESTV} waveguides. In this paper we will deal with a system akin to the last named example, let us therefore recall it in more detail. If we have two adjacent waveguides, i.e. straight planar strips in which the particle dynamics is described by (a multiple of the) Dirichlet Laplacian, the spectrum is absolutely continuous and can be found trivially by separation of variables. It is sufficient, however, to connect the two strips by a opening window in the common boundary, and the spectrum changes: a finite number of isolated eigenvalues appears under the continuum threshold despite the fact that the phase space volume corresponding to classically trapped trajectories is of measure zero. In addition, the system exhibits also resonance associated with higher transverse modes \cite{ESTV}.

In this paper we are going to discuss a different kind of resonances. We suppose that the barriers separating the two waveguides are not semiinfinite but of finite lengths $\ell_\pm$, cf.~Fig.~\ref{fig1}. Consequently, at their far ends the guide is not divided and the essential spectrum threshold is lowered to the lowest transverse-mode energy of the joint guide. The memory of the bound states in the vicinity of the window remains, however, in the form of \emph{resonances} which are especially pronounced when the barriers are long. Of the various possible definitions of a resonance we choose the most `classical' one, based on solution of the corresponding equation with the specific asymptotic behavior at large distances. It is not difficult to see that such resonances coincide with those defined by an exterior complex scaling. One naturally expects that they are at the same time also scattering resonances of our double waveguide system but we will not address this question here.

Our aim in this paper is to demonstrate existence of these resonances and their behavior as the barrier lengths tend to infinity. We are going to show that for all sufficiently large $\ell_\pm$ the number of resonances coincides with the number of eigenvalues in the semi-infinite barrier case, and that the complex energies associated with the resonances converge to the latter as $\ell_\pm\to \infty$. Moreover, we will find the leading term in the corresponding asymptotic expansion. It contains the peculiar factors decaying exponentially with respect to $\ell_\pm$ reminiscent of the Agmon metric \cite{Ag} in case of potential barriers. One can think of it as of a tunneling effect even if the classical particle in such a guide would not be trapped. The role of potential barriers is played by the narrow parts of the channel; in the physicist's terminology one would say that all the modes are evanescent in the considered range of energies.

For the sake of simplicity we consider here only the case with transverse mirror symmetry when the two guides have the same width. The system then naturally decouples into the even and odd component with only the former one being nontrivial \cite{ESTV}. One can thus analyze one waveguide only; the absence of barriers in the window and in the `outer' regions is described by Neumann boundary conditions. In the next section we describe the problem in technical terms and formulate our main result expressed in Theorem~\ref{th1}. The rest of the paper is devoted to its proof.

Before passing to our proper problem, let us add a comment on its broader context. One can regard the exterior broad-channel (or Neumann) parts as distant perturbations separated by $\ell_-+2a+\ell_+$ with $\ell_\pm$ considered large. There is a large number of problems of this type. A classical example is a double-well Schr\"odinger operator with the wells wide apart. Such systems have been studied extensively --- see, for instance, \cite{D}, \cite{AS}, \cite{HK}, \cite{Ha}, \cite{MS}, \cite{KS}, \cite{Kl}, \cite{KS1}, \cite{GHS} --- and there is no need to stress that the mechanism determining their spectral properties is the same tunneling as mentioned above in connection with resonances. Recently we managed to analyze a considerably more general class of of operators with distant perturbations and to prove general convergence theorems and to describe asymptotic behavior of their spectra and the resolvents --- cf.~\cite{B1}, \cite{B2}, \cite{B3}, \cite{B4}, \cite{GAM}, \cite{GAM1}, \cite{GAMBDI}. The general approach we have developed is useful here, since the technique employed in this paper is based on the main ideas put forward in the cited work.

\section{Problem setting and the main result}

Let $x=(x_1,x_2)$ be Cartesian coordinates in the plane. By $\Pi$ we denote a horizontal strip of width  $\pi$, i.e.  $\Pi:=\{x: -\infty<x_1<+\infty, 0<x_2<\pi\} \subset \mathbb{R}^2$. On its lower boundary we fix two disjoint segments of lengths $\ell_+$ and $\ell_-$ assuming that the distance between them is $2a>0\,$ being centered at the origin of coordinates so that we have $\gamma_+:=\{a<x_1<\ell_+,\: x_2=0\}$ and $\gamma_-:=\{-\ell_-<x_1<-a,\: x_2=0\}$, respectively. Having in mind the physical meaning of the model we will speak of them as of (Dirichlet) \emph{barriers}. The remaining part of the lower boundary consisting of three intervals separated by $\gamma_\pm$ we denote as $\Gamma$, and the whole upper boundary as $\gamma$ -- cf.~Fig.~\ref{fig1}. Our aim is to analyze the following boundary-value problem,
\begin{equation}\label{1}
\begin{gathered}
\left(-\Delta-\lambda\right)u=0 \quad \text{in} \;\; \Pi\,,
\quad
u=0 \;\; \text{on} \;\; \gamma_+\cup\gamma_-\cup\gamma \;\; \text{and}\;\;
\frac{\p u}{\p x_2}=0 \;\; \text{on}\;\; \Gamma\,,
\end{gathered}
\end{equation}
with the behavior at infinity prescribed as
\begin{equation}\label{11}
\begin{aligned}
&u(x)=C_+(\l)
\,\E^{\iu \sqrt{\lambda-\frac{1}{4}}\,x_1}\cos
\frac{x_2}{2}+\Odr\big(\E^{-\RE{\sqrt{\frac{9}{4}-\lambda}}\,x_1}\big)
\quad &&\text{for} \quad x_1\to+\infty\,,
\\&
u(x)=C_-(\l)\,\E^{-\iu \sqrt{\lambda-\frac{1}{4}}\,x_1}\cos
\frac{x_2}{2}+\Odr\big(\E^{\RE{\sqrt{\frac{9}{4}-\lambda}}\,x_1}\big)
\quad &&\text{for} \quad x_1\to-\infty\,.
\end{aligned}
\end{equation}
The choice of the square-root branch is fixed by the relation $\sqrt{1}=1$, furthermore $C_\pm(\l)$  are some constants and  $\lambda\in\mathbb{C}$ is the spectral parameter. The solution is understood in the generalized sense as an element of the Sobolev space $\Hloc^1(\Pi)$, however, by embedding it is infinitely differentiable up to the boundary for all $|x_1|$ large enough which makes it possible to interpret the asymptotics (\ref{11}) in the classical sense.

As we have indicated the main object of our interest in this paper are resonances of the problem (\ref{1}), (\ref{11}) situated in the lower complex halfplane in the vicinity of the segment $[\frac{1}{4},1]$ for large values of the barrier lengths $\ell_+$ and $\ell_-$, in particular, their asymptotic behavior as those lengths tend to infinity. Resonances are understood here as the values of $\lambda$, for which the problem (\ref{1}), (\ref{11}) has a nontrivial solution.

\begin{figure}
\includegraphics[scale=0.7]{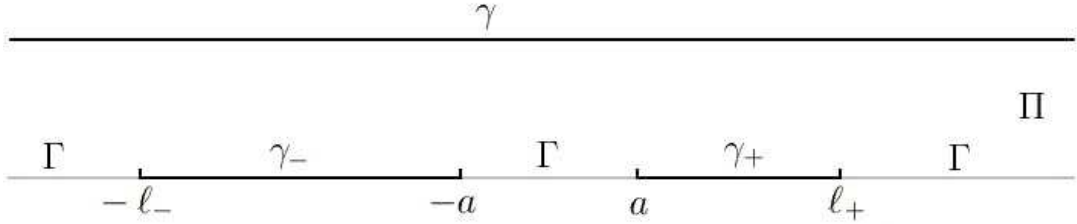}
\caption{Scheme of the model}\label{fig1}
\end{figure}

To formulate our main result we need a few more notions. Consider the segment $\Gamma_{a}:=\{|x_1|<a,\: x_2=0\}$ of length $2a$ on the lower boundary of $\Pi$ which may be called the \emph{central window}; the remaining part of the lower boundary will be denoted as $\gamma_{a}$, cf. the left picture on Fig.~\ref{fig2}. We introduce the following sesquilinear form on the space $L_2(\Pi)$,
\begin{equation*}
h_0(u,v):=\big(\nabla u, \nabla v\big)_{L_2(\Pi)}\,,
\end{equation*}
with the domain consisting of functions from $\H^1(\Pi)$ the trace of which vanishes on $\gamma_a\cup\gamma$. It is easy to see that this form is closed, symmetric, and bounded from below. The self-adjoint operator associated with it will be denoted as $\mathcal{H}_0$, in other words, $\mathcal{H}_0$ is the Laplacian on $\Pi$ with Dirichlet boundary condition on $\g\cup\g_a$ and Neumann one on $\G_a$.

It is well known \cite{ESTV, BEG, DIB2} that the operator $\mathcal{H}_0$ has for any $a>0$ a nonempty discrete spectrum consisting of a finite number of eigenvalues $\lambda_j$, $\:j=1,\ldots,n$, contained in the interval $(\frac{1}{4},1)$; without loss of generality we may assume that they are arranged in the ascending order. Each of these eigenvalues is simple and the corresponding eigenfunctions are even and odd in the variable $x_1$ for $j$ odd and even, respectively, and in the limit $x_1\to+\infty$ they behave asymptotically as
\begin{equation}\label{31}
\psi_j(x)=\Psi_j\,\E^{-\sqrt{1-\l_j}\, x_1}\sin
x_2+\Odr\big(\E^{-\sqrt{4-\l_j}\,x_1}\big)\,,
\end{equation}
where $\Psi_j$ are nonzero constants.

We shall split now the lower boundary of the strip $\Pi$ into two halflines $\g_*:=\{x: x_1<0, x_2=0\}$, $\G_*:=\{x: x_1>0, x_2=0\}$, cf. the right picture on Fig.~\ref{fig2}, and consider the boundary-value problem
\begin{equation}\label{49}
\begin{gathered}
\big(\Delta+\lambda_j\big)V_j=0\quad\text{in}\;\; \Pi\,,\quad
V_j=0 \;\; \text{on} \;\; \gamma\cup\gamma_*\,,
\quad
\frac{\p V_j}{\p x_2}=0 \;\; \text{on} \;\;\Gamma_*\,,
\end{gathered}
\end{equation}
with the following asymptotic behavior,
\begin{equation}\label{36}
\begin{aligned}
&V_j(x)=k^+_j\,\mathrm{e}^{\iu
\sqrt{\lambda_0-\frac{1}{4}}\,x_1}\cos
\frac{x_2}{2}+\Odr\big(\mathrm{e}^{- \sqrt{\frac{9}{4}-\lambda_0}\,x_1} \big)\,,
  &&   x_1\to+\infty\,,
\\
&
V_j(x)=k^-_j\,\mathrm{e}^{\sqrt{1-\lambda_0}\,x_1}\sin
x_2+\mathrm{e}^{-\sqrt{1-\lambda_0}\,x_1}\sin x_2
+\Odr\big(\mathrm{e}^{- \sqrt{4-\lambda_0}\,x_1} \big)\,,
&&  x_1\to-\infty\,,
\end{aligned}
\end{equation}
where $k_j^\pm$ are complex constants. Solution of such a problem is again understood in the generalized sense and in view of the embedding the asymptotics as $x_1\to\pm\infty$ are valid in the classical sense. In addition, it is not difficult to see that in the vicinity of the coordinate origin the function $V_j(\cdot)$ has a differentiable asymptotics, namely
\begin{equation}\label{38}
V_j(x)=\beta_j r^{\frac{1}{2}}\cos\frac{\theta}{2}+\Odr\big(r^{\frac{3}{2}}\big)\,,
\end{equation}
where $(r,\tht)$ are the appropriate polar coordinates and $\b_j$ is a complex constant. For the sake of brevity we write $\ell:=(\ell_+,\ell_-)$.

Now we are in position to state our main result.

\begin{figure}
\begin{center}
\includegraphics[scale=0.7]{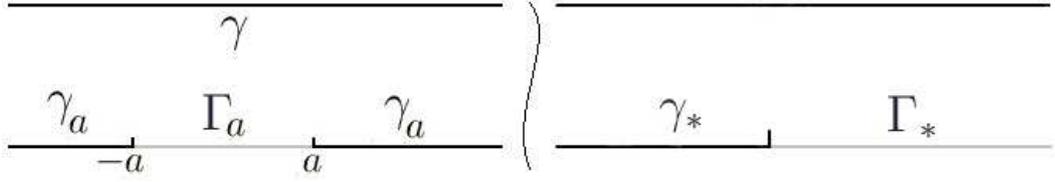}
\caption{The auxiliary problems}\label{fig2}
\end{center}
\end{figure}
\begin{theorem}\label{th1}
For all $\ell_+$ and $\ell_-$ large enough there is a unique resonance $\L_j(\ell)$ of the problem (\ref{1}), (\ref{11}) in the vicinity of each point $\lambda_j$, $\,j=1,\ldots,n$, and a unique nontrivial solution of (\ref{1}) corresponds to $\L_j$. The corresponding resonance asymptotics as \mbox{$\min\{\ell_+,\ell_-\} \to+\infty$} is given by the formula
\begin{align*}
\L_j(\ell)=&\lambda_j-k_j^-\pi
\sqrt{1-\lambda_j}|\Psi_j|^2 \Big(
\mathrm{e}^{-2\sqrt{1-\lambda_j}\,\ell_+}+
\mathrm{e}^{-2\sqrt{1-\lambda_j}\,\ell_-}\Big)
\\&
+\Odr\Big(\ell_+^2\, \mathrm{e}^{-3
\sqrt{1-\lambda_j}\,\ell_+} +\ell_-^2\, \mathrm{e}^{-3
\sqrt{1-\lambda_j}\,\ell_-}
\Big)\,.
\end{align*}
The constants $\b_j$ and $\Psi_{j}$ do not vanish and $k_j^-$ satisfies the relations
\begin{equation}\label{52}
\RE k_j^-=\frac{(\RE \beta_j)^2}{8(1-\lambda_j)}\geqslant 0\,,
\quad \IM k_j^-=\frac{(\IM \beta_j)^2}{8\sqrt{(\lambda_j-1/4)
(1-\lambda_j)}}>0\,.
\end{equation}
\end{theorem}


\section{Reduction to an operator equation}

The aim of this section is to rephrase the problem, using the ideas worked out in \cite{B1}, \cite{B2}, \cite{B3}, \cite{B4}, \cite{GAM}, \cite{GAM1}, \cite{GAMBDI}, as an operator equation; analyzing the latter we will be able to derive the leading terms in the resonance asymptotics.

Let $\widetilde{\chi}_\pm\in C^{\infty}(\mathbb{R})$ be a nonnegative cut-off function satisfying the relations
\begin{align*}
&\widetilde{\chi}_0^+(x_1)=0 \quad \text{for} \quad
x_1\geqslant 1, \quad &&\widetilde{\chi}_0^+(x_1)=1
\quad \text{for} \quad x_1\leqslant 0\,,
\\&
\widetilde{\chi}_0^-(x_1)=0 \quad \text{for} \quad
x_1\leqslant 0, \quad &&\widetilde{\chi}_0^-(x_1)=1
\quad \text{for} \quad x_1\geqslant 1\,.
\end{align*}
Furthermore, we consider nonnegative cut-off functions $\chi_0^\pm$, $\chi_\pm\in C^{\infty}(\mathbb{R})$ such that
\begin{align*}
&\chi_-(x_1)= 0 \quad \text{for} \quad x_1\geqslant -a+1\,,
\quad &&\chi_-(x_1)= 1 \quad \text{for} \quad x_1\leqslant
-a\,, \quad
\\&
\chi_+(x_1)= 0 \quad \text{for} \quad x_1\leqslant a-1\,,
\quad &&\chi_+(x_1)= 1 \quad \text{for} \quad x_1\geqslant
a\,,
\\&
\chi_0^+(x_1,\ell_+)=\widetilde{\chi}_0^+(x_1-\ell_+)\,,
\quad
&&\chi_0^-(x_1,\ell_-)=\widetilde{\chi}_0^-(x_1+\ell_-)\,.
\end{align*}
We denote by $\chi_0$ the function of the form
\begin{equation}\label{24}
\chi_0(x_1,\ell)=\chi_0^+(x_1,\ell_+)+\chi_0^-(x_1,\ell_-)-1
\end{equation}
and by $\mathcal{S}$ the shift operator acting as
\begin{equation*}
\big(\mathcal{S}(X)u\big)(x_1,x_2)=u(x_1-X,x_2)\,,
\end{equation*}
where $X$ is a real number. Let $\Pi(A,B)$ be a rectangular section of the strip, $\Pi(A,B):=\{A<x_1<B, \: 0<x_2<\pi\}$ determined by the numbers $A,\,B$.

We will introduce an auxiliary problem. We split the boundary as in (\ref{49}) and analyze the inhomogeneous boundary-value problem
\begin{equation}\label{222}
\begin{gathered}
\left(-\Delta-\lambda\right)u=g \quad \text{in} \;\; \Pi\,,
\quad
u=0 \;\; \text{on} \;\; \gamma\cup\gamma_*\,, \quad \frac{\p u}{\p x_2}=0 \;\;
\text{on} \;\; \Gamma_*\,,
\end{gathered}
\end{equation}
with the asymptotic behavior changed from (\ref{36}) to
\begin{equation}\label{2222}
\begin{aligned}
&u(x)=\widehat{C}_+(g,\l)
\,\mathrm{e}^{\iu
\sqrt{\lambda-\frac{1}{4}}\,x_1}\cos
\frac{x_2}{2}+\Odr\big(\mathrm{e}^{-\RE{\sqrt{\frac{9}{4}-\lambda}}\,x_1}\big)
\quad &&\text{for} \quad x_1\to+\infty\,,
\\&
u(x)=\widehat{C}_-(g,\l)\,\mathrm{e}^{\sqrt{1-\lambda}\,x_1}\sin
x_2+\Odr\big(\mathrm{e}^{-\RE{\sqrt{4-\lambda}}\,x_1}\big) \quad
&&\text{for} \quad x_1\to-\infty\,.
\end{aligned}
\end{equation}
We assume here that $g\in L_2(\Pi)$ is a function satisfying $\supp g \subseteq \Pi(-1,1)$ and $\widehat{C}_\pm(g,\l)$ are complex constants. Solutions of this problem are understood in the same sense as the solutions of (\ref{49}), (\ref{36}).

We shall seek solutions to equation (\ref{1}) in the form
\begin{equation}\label{3}
\begin{aligned}
u_{\ell}(x)&=\chi_0(x_1,\ell)u_0(x)
+\chi_+(x_1)\Big(\mathcal{S}(\ell_+)u_+\Big)(x)
+\chi_-(x_1)\Big(\mathcal{S}(-\ell_-)u_-\Big)(x)\,,
\end{aligned}
\end{equation}
where $u_0:=(\mathcal{H}_0-\lambda)^{-1}g_0$ with a function $g_0\in L_2(\Pi)$ having the property that $\supp g_0\subseteq \Pi(-a-1,a+1)$. By definition the function $u_0$ satisfies the equation
\begin{equation}\label{26}
\left(\mathcal{H}_0-\lambda\right)u_0=g_0
\end{equation}
and behaves asymptotically in the following way,
\begin{equation}\label{1111}
u_0(x)=C_\pm^0(g_0,\l)\,
\mathrm{e}^{\pm\sqrt{1-\lambda}\,x_1}\sin
x_2+\Odr\big(\mathrm{e}^{-\RE{\sqrt{4-\lambda}}\,x_1}\big) \quad
\text{for} \quad x_1\to\pm\infty\,,
\end{equation}
where $C_\pm^0(g_0,\l)$  are some constants. The function $u_+$ is determined as the solution to the boundary-value problem (\ref{222}), (\ref{2222}) with a right-hand side $g=g_+\in L_2(\Pi)$ satisfying $\supp g_+\subseteq \overline{\Pi(-1,1)}$. The function $u_-$ is introduced in a bit more complicated way, namely we suppose that its mirror image in the $x_1$ variable, $u_-(-x_1,x_2)$, solves the problem (\ref{222}), (\ref{2222}) with a right-hand side $g(x)=g_-(-x_1,x_2)$, where $g_-\in L_2(\Pi)$ with $\supp g_-\subseteq \overline{\Pi(-1,1)}$.

In view of the definition of $\chi_0$, $\chi_\pm$, $u_0$, $u_\pm$ the function $u_{\ell}$ given by (\ref{3}) satisfies the boundary conditions of the problem (\ref{1}) and has the needed asymptotic behavior (\ref{11}). It remains to check that $u_{\ell}$ solves the equation (\ref{1}). To this aim, we substitute the Ansatz (\ref{3}) into it obtaining
\begin{align*}
\big(-\Delta-\lambda\big)\Big(\chi_0(x_1,\ell)u_0(x)
+\chi_+(x_1)\big(\mathcal{S}(\ell_+)u_+\big)(x)
+\chi_-(x_1)\big(\mathcal{S}(-\ell_-)u_-\big)(x)\Big)=0\,.
\end{align*}
Taking into account the  cut-off functions definitions, the problem (\ref{222}) and the relation (\ref{24}), we arrive at the equations
\begin{equation}\label{5}
\begin{aligned}
&g_0+\mathcal{T}_+g_+
+\mathcal{T}_-g_-=0\,,
\\&
g_\pm+\widetilde{\mathcal{T}}_\pm g_0=0\,,
\end{aligned}
\end{equation}
where
\begin{equation}\label{4}
\begin{aligned}
&\mathcal{T}_\pm g_\pm=U_\pm\frac{d^2\chi_\pm}{d x_1^2}
+2\frac{d \chi_\pm}{d x_1} \frac{\p U_\pm}{\p x_1}\,, \quad
&&U_\pm=\mathcal{S}(\pm\ell_\pm)u_\pm\,,
\\&
\widetilde{\mathcal{T}}_\pm g_0=-\mathcal{S}(\mp\ell_\pm)\Big(u_0\frac{d^2 \chi_0^\pm}{d x_1^2}
+2\frac{\p u_0}{\p x_1}\frac{d \chi_0^\pm}{d x_1}\Big)\,.
\end{aligned}
\end{equation}
The operators $\mathcal{T}_\pm=\mathcal{T}_\pm(\l,\ell)$ map from $L_2(\Pi(-a-1,a+1))$, and $\widetilde{\mathcal{T}}_\pm =\widetilde{\mathcal{T}}_\pm(\l,\ell)$ from $L_2(\Pi(-a-1,a+1))$ into $L_2(\Pi(-1,1))$; it is obvious from the construction that all of them are bounded.

Next we are going to formulate a number of auxiliary result which we shall need to analyze the equations (\ref{5}). For simplicity we will denote in the following by $\lambda_0$ a fixed eigenvalue $\lambda_j$ of the operator $\mathcal{H}_0$ and $\psi_0$ will be the corresponding eigenfunction $\psi_j$.

\begin{lemma}\label{lm0}
The eigenfunction $\psi_0$ of $\mathcal{H}_0$ can be represented as a uniformly con\-vergent series,
\begin{equation}\label{29}
\psi_0(x)=\sum\limits_{n=1}^{\infty}
\Psi_{0,n}^\pm\,\mathrm{e}^{\mp\sqrt{n^2-\lambda_0}\, x_1}
\sin nx_2 \quad \text{for} \quad \pm x_1>a\,,
\end{equation}
where the convergence is understood in the sense of the norm of $\H^1\big(\Pi\backslash\Pi(-a,a)\big)$. The coefficients $\Psi_{0,n}^\pm$ satisfy the relations
\begin{equation}\label{28}
\sum\limits_{i=1}^{\infty}n|\Psi_{0,n}^\pm|^2\leqslant C
\|\psi_0\|_{\H^1\left(\Pi(-a,a)\right)}^2
\end{equation}
with some constant $C$.
\end{lemma}

It is easy to check the lemma using separation of variables.

\begin{lemma}\label{lm1}
For any $f\in L_2(\Pi)$ with $\supp f \subseteq \Pi(-a-1,a+1)$ and all $\lambda$ in the vicinity of $\lambda_0$ we have the representation
\begin{equation}\label{10}
\big(\mathcal{H}_0-\lambda\big)^{-1}f=
\frac{(f,\psi_0)_{L_2(\Pi)}}{\lambda_0-\lambda}\,\psi_0
+\mathcal{R}_0(\lambda)f\,,
\end{equation}
where $\mathcal{R}_0(\lambda)$ is the reduced resolvent in the sense of Kato, holomorphic in $\lambda$ in the vicinity of $\lambda_0$ and acting in the orthogonal complement to the eigenfunction $\psi_0$. Moreover, we have the estimate
\begin{equation}\label{12}
\|\mathcal{R}_0(\lambda)f\|_{\H^2(\Pi\backslash\Pi(-b,b))}\leqslant
C b \, \mathrm{e}^{-\varrho b}\|f\|_{L_2(\Pi(-a-1,a+1))}\,,
\end{equation}
where $\varrho=\min\Big\{\sqrt{1-\lambda_0}, \RE\sqrt{1-\lambda}\Big\}$, $\:b\geqslant a+1$, and  $C$ is a constant independent of $f$, $\lambda$ and $b$.
\end{lemma}
\begin{proof}
Validity of the expansion (\ref{10}) follows from formula (3.21) in \cite[Sec.~V.3.5]{K}. Let us check the estimate (\ref{12}). To this aim, we denote $w=\mathcal{R}_0(\lambda)f\in\H^2\big(\Pi\backslash\Pi(-b,b)\big)$. From (\ref{26}), (\ref{1111}) and the definition of the function $f$ we derive an equation which should be satisfied by the function $w$,
\begin{equation*}
\begin{gathered}
(-\Delta-\lambda)w=K\psi_0 \quad \text{in}\quad \Pi\backslash \Pi(-a-1,a+1)\,,
\quad  K:=-(f,\psi_0)_{L_2(\Pi)}\,.
\end{gathered}
\end{equation*}
Substituting now the expansion (\ref{29}) for the eigenfunction $\psi_0$ from Lemma~\ref{lm0}, we get
\begin{align*}
&\left((-\Delta-\lambda)w\right)(x_1,x_2) =
K\sum\limits_{n=1}^{\infty}\Psi_{0,n}^\pm\,
\mathrm{e}^{\mp\sqrt{1-\lambda_0}\,x_1
}\sin
nx_2 \quad \text{for} \quad \pm x_1>a+1\,.
\end{align*}
Solutions to the last equation are found easily using separation of variables. They are of the form $w=\mathrm{P}_1 +\mathrm{P}_2$, where
\begin{align*}
&\mathrm{P}_1=K\sum\limits_{n=1}^{\infty}\frac{\Psi_{0,n}^\pm\, \E^{-\sqrt{n^2-\l_0}\,a}}{\lambda_0-\lambda}
\Big(\mathrm{e}^{\mp\sqrt{n^2-\lambda_0}(x_1\mp a)}-\mathrm{e}^{\mp\sqrt{n^2-\lambda}(x_1\mp a)}\Big)\sin
n x_2, &&  \pm x_1>a+1,
\\
&\mathrm{\mathrm{P}}_2=\sum\limits_{n=1}^{\infty}W_{0,n}^\pm
\mathrm{e}^{\mp\sqrt{n^2-\lambda_0}(x_1\mp a)}\sin
nx_2, && \pm x_1>a+1,
\end{align*}
and $W_{0,n}^\pm$ are the Fourier coefficients of $w(\pm a, x_2)$, respectively. We have the estimates
\begin{equation}\label{32}
\sum\limits_{n=1}^{\infty}n|W_{0,n}^\pm|^2\leqslant C
\|w\|_{\H^1\left(\Pi(-b,b)\right)}^2.
\end{equation}
Here and in the remaining part of the proof we will denote by $C$ unspecified constants independent of $\lambda$, $n$ and $f$.

Let us check the estimate (\ref{12}) for $x_1>b$. Writing $x=(x_1,x_2)$ it is straight\-forward to check the inequality
\begin{equation}\label{44}
\|w\|_{L_2\left(\Pi(b,+\infty)\right)}^2\leqslant
2\int\limits_{\Pi(b,+\infty)}\big|\mathrm{P}_1(x)\big|^2\di x
+2\int\limits_{\Pi(b,+\infty)}\big|\mathrm{P}_2(x)\big|^2\di x\,.
\end{equation}
Let us evaluate the first integral,
\begin{align}\label{19}
\int\limits_{\Pi(b,+\infty)}\big|\mathrm{P}_1(x)\big|^2\di x
=\frac{\pi}{2}K^2\sum\limits_{n=1}^{\infty}\frac{|\Psi_{0,n}^+|^2\, \E^{-2\sqrt{n^2-\l_0}\,a}}
{|\lambda-\lambda_0|^2}\,
\big(\mathrm{N}_1+\mathrm{N}_2\big)\,,
\end{align}
where
\begin{equation}\label{23}
\begin{aligned}
&\mathrm{N}_1:=\int\limits_{b}^{+\infty}
\mathrm{e}^{-2\RE\sqrt{n^2-\lambda}\,x_1}
\sin^2\IM\sqrt{n^2-\lambda}\,x_1\:\di x_1\,,
\\&
\mathrm{N}_2:=\int\limits_{b}^{+\infty}\mathrm{e}^{-2\sqrt{n^2-\lambda_0}\,x_1}
\Big(1
-\mathrm{e}^{-\big(\RE\sqrt{n^2-\lambda}-\sqrt{n^2-\lambda_0}\big)\,x_1}
\cos\IM\sqrt{n^2-\lambda}x_1\Big)^2\di x_1\,.
\end{aligned}
\end{equation}
Using now the representation
\begin{equation*}
\begin{aligned}
1-\mathrm{e}^{-\mu x}\cos\nu x=\big(1-\mathrm{e}^{-\mu x}\big)
+\mathrm{e}^{-\mu x}\big(1-\cos\nu x\big) \quad \text{with} \quad \mu, \nu\in \mathbb{R}\,,
\end{aligned}
\end{equation*}
we derive an estimate to the second integral in (\ref{23}), namely
\begin{equation*}
\mathrm{N}_2
\leqslant 2\mathrm{Q}_1+
2\mathrm{Q}_2\,,
\end{equation*}
where
\begin{align*}
&\mathrm{Q}_1:=\int\limits_{b}^{+\infty}\mathrm{e}^{-2\sqrt{n^2-\lambda_0}\,x_1}
\Big(1
-\mathrm{e}^{-\big(\RE\sqrt{n^2-\lambda}-\sqrt{n^2-\lambda_0}\big)x_1}\Big)^2\di x_1,
\\&
\mathrm{Q}_2:=\int\limits_{b}^{+\infty}\mathrm{e}^{-2\RE\sqrt{n^2-\lambda}\,x_1}
\sin^{4}\frac{\IM\sqrt{n^2-\lambda}\,x_1}{2}\,\di x_1.
\end{align*}
Using the real and imaginary parts of the identity
\begin{equation*}
\begin{gathered}
\sqrt{n^2-\lambda}-\sqrt{n^2-\lambda_0}
=\frac{\lambda-\lambda_0}{\sqrt{n^2-\lambda}+\sqrt{n^2-\lambda_0}}
\end{gathered}
\end{equation*}
and taking into account the considered range of $\lambda$, we get the inequality
\begin{align*}
&\max\left\{ \RE\big(\sqrt{n^2-\lambda}-\sqrt{n^2-\lambda_0}\big),\:
\IM\big(\sqrt{n^2-\lambda}-\sqrt{n^2-\lambda_0}\big) \right\}
\leqslant
C\,\frac{|\lambda-\lambda_0|}{n}\,.
\end{align*}
Combining it with the estimate
\begin{equation*}
1 -\mathrm{e}^{-t}\leqslant
 t\, \mathrm{e}^{|t|},\quad t\in \mathbb{R}\,,
\end{equation*}
we infer that
\begin{equation}\label{27}
\max\{\mathrm{N}_1,\, \mathrm{Q}_1,\, \mathrm{Q}_2\} \leqslant C b^2\,\mathrm{e}^{-2\RE\sqrt{n^2-\lambda}\,b}\,\frac{|\lambda-\lambda_0|^2}
{n^2}\,.
\end{equation}
Substituting the last estimate into (\ref{19}) and taking into account the bound (\ref{28}) for the coefficients $\Psi_{0,n}^+$ from Lemma~\ref{lm0} together with the relation $K=-(f,\psi_0)_{L_2(\Pi)}$, we find
\begin{equation}\label{34}
\int\limits_{\Pi(b,+\infty)}\big|\mathrm{P}_1(x)\big|^2\di x\leqslant C b^2\,
\mathrm{e}^{-2\varrho b}\|f\|_{L_2(\Pi(-a-1,a+1))}^2\,.
\end{equation}
Arguing in a similar way and using the inequalities (\ref{32}), (\ref{27}) in combination with the relation $w=\mathcal{R}_0(\lambda)f$, we derive a bound for the second integral in (\ref{44}),
\begin{equation*}
\int\limits_{\Pi(b,+\infty)}\big|\mathrm{P}_2(x)\big|^2\di x
\leqslant C b^2\, \mathrm{e}^{-2b\sqrt{1-\lambda_0}}\|f\|_{L_2(\Pi(-a-1,a+1))}^2\,.
\end{equation*}
The last result together with (\ref{34}) yields the estimate
\begin{equation*}
\|w\|_{L_2(\Pi(b,+\infty))}\leqslant
C b\, \mathrm{e}^{-\varrho b}\|f\|_{L_2(\Pi(-a-1,a+1))}\,.
\end{equation*}
Calculation the first and the second derivatives of the function $W$ and using the analogous reasoning, it is straightforward to check the estimates
\begin{equation*}
\begin{aligned}
\|w\|_{\H^1(\Pi(b,+\infty))}\leqslant
C b\, \mathrm{e}^{-\varrho b}\|f\|_{L_2(\Pi(-a-1,a+1))}\,,
\\
\|w\|_{\H^2(\Pi(b,+\infty))}\leqslant
C b\, \mathrm{e}^{-\varrho b}\|f\|_{L_2(\Pi(-a-1,a+1))}\,.
\end{aligned}
\end{equation*}
This implies the sought estimate (\ref{32}) for the norm $\|\mathcal{R}_0(\lambda)f\|_{\H^2(\Pi(b,+\infty))}$.
The bound for $\|\mathcal{R}_0(\lambda)f\|_{\H^2(\Pi(-\infty,b))}$ is checked in the same way.
\end{proof}
\begin{lemma}\label{lm2}
For any function $f\in L_2(\Pi)$ with $\supp f \subseteq \Pi(-1,1)$ and all $\lambda$ in the vicinity of $\lambda_0$,
the operators $\widetilde{\mathcal{T}}_\pm$ are of the form
\begin{equation*}
\begin{aligned}
&\widetilde{\mathcal{T}}_\pm(\l,\ell)f=
-\frac{(f,\psi_0)_{L_2(\Pi)}}{\lambda_0-\lambda}\, \mathcal{S}(\mp\ell_\pm)
\Big(\psi_0\frac{d^2\chi_0^+}{d x_1^2}
+2\frac{\p \psi_0}{\p x_1}\frac{d \chi_0^+}{d
x_1}\Big)
-\mathcal{R}_\pm(\lambda,\ell)f\,,
\end{aligned}
\end{equation*}
where the operators $\mathcal{R}_\pm(\lambda,\ell): L_2(\Pi(-a-1,a+1))\to L_2(\Pi(-1,1))$ are holomorphic as functions of $\lambda$ in the vicinity of $\lambda_0$. Moreover, the inequality
\begin{equation*}
\begin{aligned}
&\big\|\mathcal{R}_\pm(\lambda,\ell) f\big\|_{\H^2(\Pi(-1,1))}
\leqslant C \ell_\pm\, \mathrm{e}^{-\varrho
\ell_\pm}\|f\|_{L_2(\Pi(-1,1) )},
\end{aligned}
\end{equation*}
holds true with a constant $C$ independent of $f$, $\ell_\pm$, and  $\lambda$.
\end{lemma}
The claim of the lemma follows directly from the representation (\ref{4}) and Lemma~\ref{lm1}. We shall need also a statement concerning embedded eigenvalues.

\begin{lemma}\label{lm6}
For any $\l\in[1/4,1]$ the boundary-value problem (\ref{222}), (\ref{2222}) with the vanishing right-hand side has only a trivial solution.
\end{lemma}
\begin{proof}
Consider the rectangle $\Pi(-R_1,R_2)$ with fixed $R_1,\,R_2$. We multiply equation (\ref{222}) by $x_1\overline{u}$ and integrate it twice by parts in the region $\Pi(-R_1,R_2)$ taking into account the boundary conditions imposed on the function and its behavior at zero. It yields
\begin{equation}\label{16}
\begin{aligned}
0&=\int\limits_{\Pi(-R_1,R_2)}x_1\overline{u}\big(-\Delta-\lambda\big)
\frac{\p u}{\p x_1}\di x
=\int\limits_{\Pi(-R_1,R_2)}\frac{\p u}{\p x_1}(-\Delta-\lambda)
x_1\overline{u}\di x
\\&
+\int\limits_0^{\pi}\frac{\p u}{\p x_1}\Big|_{x_1=R_2}\frac{\p }
{\p x_1}x_1\overline{u}
\Big|_{x_1
=R_2}\di x_2
+R_1\int\limits_0^{\pi}\overline{u}(-R_1,x_2)
\frac{\p^2 u}{\p x_1^2}\Big|_{x_1=-R_1}\di x_2
\\&
+\int\limits_0^{\pi}\frac{\p u}{\p x_1}\Big|_{x_1=-R_1}
\frac{\p}{\p x_1}x_1\overline{u}\Big|_{x_1=-R_1}\di x_2
-R_2\int\limits_0^{\pi}\overline{u}(R_2,x_2)
\frac{\p^2 u}{\p x_1^2}\Big|_{x_1=R_2}\di x_2\,.
\end{aligned}
\end{equation}
Using separation of variables we represent the function $u$ for $x_1>0$ as
\begin{equation}\label{18}
\begin{aligned}
&u(x)=\widehat{C}_+\,\mathrm{e}^{\iu
\sqrt{\lambda-\frac{1}{4}}x_1}\cos
\frac{x_2}{2}+\widetilde{u}_+(x)\,,
\end{aligned}
\end{equation}
where the function $\widetilde{u}(x)$ decays exponentially and satisfies the relation
\begin{equation*}
\int\limits_0^{\pi}\widetilde{u}(x_1,x_2)\cos\frac{x_2}{2}\,\di x_2=0\,.
\end{equation*}
Next we find the limit of the expression (\ref{16}) as $R_1\to+\infty$ taking into account equation (\ref{222}), relation (\ref{18}), boundary conditions imposed on the function $u$ and its asymptotic behavior; we get
\begin{equation}\label{54}
\begin{aligned}
0=&\int\limits_{\Pi(-\infty,R_2)}x_1\overline{u}(-\Delta-\lambda)
\frac{\p u}{\p x_1}\di x
=-2\int\limits_{\Pi(-\infty,R_2)}
\Big|\frac{\p u}{\p x_1}\Big|^2\di x
\\+
&\int\limits_0^{\pi}\frac{\p u}{\p x_1} \frac{\p}{\p
x_1}x_1 \overline{u}\Big|_{x_1=R_2}\di x_2
-R_2\int\limits_0^{\pi}\overline{u}(R_2,x_2) \frac{\p^2 u}
{\p x_1^2} \Big|_{x_1=R_2}\di x_2
\\
=&\,\iu\frac{\pi}{2}\big|\widehat{C}_+\big|^2\sqrt{\lambda-\frac{1}{4}}
+\pi R_2\big|\widehat{C}_+\big|^2\Big(\lambda-\frac{1}{4}\Big)
\\
&
-2\int\limits_{\Pi(-\infty,R_2)}
\Big|\frac{\p u}{\p x_1}\Big|^2\di x
+\int\limits_0^{\pi}\left(\frac{\p \widetilde{u}}{\p x_1}
\frac{\p}{\p x_1}x_1 \overline{\widetilde{u}} -x_1 \overline{u}(\widetilde{x}) \frac{\p^2 \widetilde{u}}
{\p x_1^2}\right)\Big|_{x_1=R_2} \di x_2\,.
\end{aligned}
\end{equation}
It remains only to evaluate the integral over the region $\Pi(-\infty,R_2)$. To this aim we consider the region $\Pi(-\infty,0)$. Taking into account relation (\ref{18}) and the definition of the two regions, we find
\begin{align*}
&-2\!\!\!\!\!\int\limits_{\Pi(-\infty,R_2)}\Big|\frac{\p u}{\p x_1}\Big|^2\di x
=-\,2\!\!\!\int\limits_{\Pi(-\infty,0)}\Big|\frac{\p u}{\p x_1}\Big|^2\di x
-2\!\!\!\!\!\!\!\!\!\!\!\!\int\limits_{\Pi(-\infty,R_2)\backslash\Pi(-\infty,0)}\Big|
\frac{\p u}{\p x_1}\Big|^2\di x
\\&
=-\,2\!\!\!\int\limits_{\Pi(-\infty,0)}\Big|\frac{\p u}{\p x_1}\Big|^2\di x
-2\!\!\!\!\!\!\!\!\!\!\!\!\int\limits_{\Pi(-\infty,R_2)\backslash\Pi(-\infty,0)}\Big|
\frac{\p \widetilde{u}}{\p x_1}\Big|^2\di x
-\pi\big|\widehat{C}_+\big|^2\Big(\lambda-\frac{1}{4}\Big)R_2
\end{align*}
We substitute the last identity into (\ref{54}) and evaluating its limit as $R_2\to +\infty$, we get
\begin{equation*}
0=-\,2\!\!\!\int\limits_{\Pi(-\infty,0)}\Big|\frac{\p u}{\p x_1}\Big|^2\di x
-2\!\!\!\!\!\!\!\!\!\!\int\limits_{\Pi(-\infty,\infty)\backslash\Pi(-\infty,0)}\Big|
\frac{\p \widetilde{u}}{\p x_1}\Big|^2\di x
+\iu\frac{\pi}{2}\big|\widehat{C}_+\big|^2\sqrt{\lambda-\frac{1}{4}}\,.
\end{equation*}
This is equivalent to the following two relations
\begin{equation*}
\begin{aligned}
&2\!\!\!\int\limits_{\Pi(-\infty,0)}\Big|\frac{\p u}{\p x_1}\Big|^2\di x
+2\!\!\!\!\!\!\!\!\int\limits_{\Pi(-\infty,\infty)\backslash\Pi(-\infty,0)}\Big|
\frac{\p \widetilde{u}}{\p x_1}\Big|^2\di x=0\,,
\quad
\frac{\pi}{2}\big|\widehat{C}_+\big|^2\sqrt{\lambda-\frac{1}{4}}=0\,,
\end{aligned}
\end{equation*}
which imply that $u(x)$ is independent of $x$ and $\widehat{C}_+=0$; this allows us to conclude that the function $u(\cdot)$ is equal to zero identically.
\end{proof}

The idea of the proof of the last lemma is borrowed from Lemma~3.3 in \cite{DIB1}.

\begin{lemma}\label{lm3}
For any $f\in L_2(\Pi)$ with $\supp f \subseteq \Pi(-1,1)$ the operators $\mathcal{T}_\pm(\l,\ell)$ are uniformly bounded and satisfy the inequalities
\begin{equation}\label{37}
\begin{aligned}
&\|\mathcal{T}_\pm(\l,\ell)f\|_{\H^2\left(\Pi(-1,1)\right)}\leqslant
C \ell_\pm\, \mathrm{e}^{-\varrho
\ell_\pm}\|f\|_{L_2\left(\Pi(-1,1)\right)}
\end{aligned}
\end{equation}
with a constant $C$ independent of $f$, $\lambda$, and $\ell_\pm$.
\end{lemma}
\begin{proof}
The operator boundedness is checked by Lemma~\ref{lm6} in analogy with Lemmata~5.2, 5.3 of the paper \cite{DIB2}, and the bounds (\ref{37}) are derived using separation of variables and transformations we have used above in the proof of Lemma~\ref{lm1}.
\end{proof}

Let us return now to discussion of equations (\ref{5}). We introduce the Hilbert space
\begin{align*}
\mathfrak{L}:=\Bigg\{h=
\begin{pmatrix}
h_1
\\
h_2
\\
h_{3}
\end{pmatrix}\!,\:
h_{i}\in L_2(\mathbb{R}^{d};\mathbb{C}^{n}), \
i=1,2,3\Bigg\}
\end{align*}
with the scalar product
\begin{align*}
(u,v)_{\mathfrak{L}}=\sum\limits_{i=1}^{3}(u_{i},v_{i})_{L_2(\mathbb{R}^2)}\,.
\end{align*}
We denote
\begin{align*}
\begin{gathered}
g:=
\begin{pmatrix}
g_0
\\
g_+
\\
g_-
\end{pmatrix} \in \mathfrak{L}\,.
\end{gathered}
\end{align*}
Having introduced $\mathfrak{L}$ we consider the operator
\begin{equation*}
\mathcal{T}(\l,\ell):=
\begin{pmatrix}
0 & \mathcal{T}_+(\l,\ell) & \mathcal{T}_-(\l,\ell)
\\
\widetilde{\mathcal{T}}_+(\l,\ell) & 0 &  0
\\
\widetilde{\mathcal{T}}_-(\l,\ell) & 0 &  0
\end{pmatrix}
\end{equation*}
on this space. Using this notation we can rewrite equations (\ref{5}) as
\begin{equation*}
g+\mathcal{T}(\l,\ell)g=0
\end{equation*}
interpreting this relation as an operator equation in $\mathfrak{L}$. It can be analyzed using the appropriate version of Birman-Schwinger method described in the papers \cite{GRR, DIB2}, which makes it possible to rephrase the search for resonances of the problem (\ref{1}), (\ref{11}) to finding zeros of an explicitly given function.

In accordance with Lemma~\ref{lm2} we have the representation
\begin{equation}\label{7}
g-\frac{\big(g,\phi\big)_{\mathfrak{L}}}{\lambda_0-\lambda}\Psi_0
+\mathcal{R}(\lambda,\ell)g=0
\end{equation}
where
\begin{align*}
\begin{gathered}
\phi:=
\begin{pmatrix}
\psi_0
\\
0
\\
0
\end{pmatrix},
\quad \Psi_0:=
\begin{pmatrix}
0
\\
\mathcal{S}(-\ell_+)\big(\psi_0\frac{d^2\chi_0^+}{d x_1^2}
+2\frac{\p \psi_0}{\p x_1}
\frac{d \chi_0^+}{d x_1}\big)
\\
\mathcal{S}(\ell_-)\big(\psi_0\frac{d^2 \chi_0^-}{d x_1^2}
+2\frac{\p \psi_0}{\p x_1}
\frac{d \chi_0^-}{d x_1}\big)
\end{pmatrix},
\\
\mathcal{R}(\lambda,\ell):=
\begin{pmatrix}
0 & -\mathcal{T}_+ & -\mathcal{T}_-
\\
\mathcal{R}_+(\lambda,\ell) & 0 &  0
\\
\mathcal{R}_-(\lambda,\ell) & 0 &  0
\end{pmatrix} \in \mathfrak{L}\,.
\end{gathered}
\end{align*}
We note that, as a consequence of the definition of the vector $\Psi_0$ --- recall that $\psi_0$ is exponentially decaying
by Lemma~\ref{lm0}, we have the estimate
\begin{equation}\label{3.23}
\|\Psi_0\|_{\mathfrak{L}}=\Odr\big( \mathrm{e}^{-2\sqrt{1-\lambda_0}\,\ell_+}+
\mathrm{e}^{-2\sqrt{1-\lambda_0}\,\ell_-}
\big)\,.
\end{equation}
Putting together the first and the last term in (\ref{7}), we get
\begin{equation}\label{9}
\big(\I+\mathcal{R}(\lambda,\ell)\big)g-
\frac{\big(g,\phi\big)_{\mathfrak{L}}}{\lambda_0-\lambda}\Psi_0=0\,,
\end{equation}
where $\I$ is the unit operator. To analyze this equation, we need one more lemma; recall that we have introduced $\varrho=\min\Big\{\sqrt{1-\lambda_0}, \RE\sqrt{1-\lambda}\Big\}$.

\begin{lemma}\label{lm4}
For any $f\in\mathfrak{L}$ and $\lambda$ in the vicinity of $\lambda_0$ we have
\begin{equation*}
\big\|\mathcal{R}(\lambda,\ell)f\big\|_{\mathfrak{L}}\leqslant C
\Big(\ell_+\mathrm{e}^{-\varrho \ell_+}
+\ell_-\mathrm{e}^{-\varrho \ell_-}\Big)\|f\|_{\mathfrak{L}}
\end{equation*}
with a constant $C$ independent of $f$, $\lambda$, and $\ell_\pm$.
\end{lemma}
The claim follows easily from Lemmata~\ref{lm2} and \ref{lm3}.

By the last lemma, the inverse $\big(\I+\mathcal{R}(\lambda,\ell)\big)^{-1}$ exists for $\ell_\pm$ large enough. In such a case we can apply it to both sides of equation (\ref{9}) obtaining
\begin{equation*}
g-\frac{\big(g,\phi\big)_{\mathfrak{L}}}{\lambda_0-\lambda}
\big(\I+\mathcal{R}(\lambda,\ell)\big)^{-1}\Psi_0=0\,.
\end{equation*}
Taking the scalar product of the two terms with the function $\phi$ defined above we get
\begin{equation*}
(g,\phi)_{\mathfrak{L}}-\frac{\big(g,\phi\big)_{\mathfrak{L}}}{\lambda_0-\lambda}
\Big(\big(\I+\mathcal{R}(\lambda,\ell)\big)^{-1}\Psi_0,\phi\Big)_{\mathfrak{L}}
=0\,,
\end{equation*}
or equivalently
\begin{equation*}
(g,\phi)_{\mathfrak{L}}\Big(\I-\frac{1}{\lambda_0-\lambda}
\Big(\big(\I+\mathcal{R}(\lambda,\ell)\big)^{-1}\Psi_0,\phi\Big)_{\mathfrak{L}}\Big)
=0\,.
\end{equation*}
The scalar product $(g,\phi)_{\mathfrak{L}}$ does not vanish, since otherwise we would have $g\equiv0$ and the eigenvalues $\lambda$ of the perturbed operator would be zero. We may thus  cancel the
 factor $(g,\phi)_{\mathfrak{L}}$ obtaining
\begin{equation}\label{8}
\lambda=\lambda_0-\Big(\big(\I+\mathcal{R}(\lambda,\ell)\big)^{-1}
\Psi_0,\phi\Big)_{\mathfrak{L}}\,.
\end{equation}

\section{The resonance asymptotics}

The aim of this section is, using the preliminaries discussed above, to prove that equation (\ref{8}) locally has  a unique solution and to construct the leading term in the resonance asymptotics of the perturbed problem (\ref{1}), (\ref{11}).

\begin{lemma}\label{lm5}
Equation (\ref{8}) has a unique solution.
\end{lemma}
\begin{proof}
Using $z=\lambda-\lambda_0$ in equation (\ref{8}) we get
\begin{equation*}
F(z):=z+
G(z)=0\,,
\end{equation*}
where
\begin{equation*}
G(z):=\Big(\big(\I+\mathcal{R}(\l_0+z,\ell)\big)^{-1}
\Psi_0,\phi\Big)_{\mathfrak{L}}\,.
\end{equation*}
The function $z\mapsto G(z)$ is analytic $|z|<z_0$, where $z_0$ is a sufficiently small positive number. In the absence of $G(z)$ we have at the left-hand side the identical function $z\mapsto z$ which has a unique simple zero at $z=0$. We are going to show that $G(z)$ is a small perturbation. It follows from (\ref{3.23}) and Lemma~\ref{lm4} that it satisfies the bound
\begin{equation*}
|G(z)|\leqslant C \big( \mathrm{e}^{-2\sqrt{1-\lambda_0}\,\ell_+}+
\mathrm{e}^{-2\sqrt{1-\lambda_0}\,\ell_-}
\big) \quad \text{for all} \quad |z|\leqslant z_0
\end{equation*}
with a constant $C$ independent of $z$ and $\ell_\pm$. It follows that
\begin{equation*}
|G(z)|
\leqslant C\Big(\ell_+\mathrm{e}^{-\varrho\ell_+}
+\ell_-\mathrm{e}^{-\varrho\ell_-}\Big)
< z_0
\end{equation*}
holds for $|z|=z_0$ as $\ell_\pm\to+\infty$ which makes it possible to apply to the  function $F(z)$ Rouch\'{e}'s theorem --- cf.~\cite[Sec.~IV.3]{M} --- by which it has just one simple root in the circle $|z|\leqslant z_0$.
\end{proof}

Equation (\ref{8}) can be rewritten in the form
\begin{equation}\label{25}
\begin{aligned}
\lambda-\lambda_0&=-\Big(\Big(\I-\mathcal{R}(\lambda,\ell) +\mathcal{R}^2(\lambda,\ell)\big
(\I+\mathcal{R}(\lambda,\ell)\big)^{-1}\Big)
\Psi_0,\phi\Big)_{\mathfrak{L}}
\\&
=-\big(\Psi_0,\phi\big)_{\mathfrak{L}} +\big(\mathcal{R}(\lambda,\ell)\Psi_0,\phi\big)_{\mathfrak{L}}
-\Big(\mathcal{R}^2(\lambda,\ell)\big
(\I+\mathcal{R}(\lambda,\ell)\big)^{-1}\Psi_0,\phi\Big)_{\mathfrak{L}}\,;
\end{aligned}
\end{equation}
for the first term at the right-hand side we get $\big(\Psi_0,\phi\big)_{\mathfrak{L}}=0$ using the definitions of the functions $\Psi_0,\,\phi$ and the scalar product in the space $\mathfrak{L}$.

Let $\L=\L(\ell)$ be a root of the equation (\ref{8}). In view of the boundedness of the operator $(\I+\mathcal{R}(\lambda,\ell)\big)^{-1}$, Lemma~\ref{lm4} and the definition of the function $\phi$ we have the inequality
\begin{equation}\label{13}
\big|\L-\lambda_0\big|
\leqslant
C\|\mathcal{R}(\L,\ell)\|=\Odr(\ell_+\mathrm{e}^{-\varrho\ell_+}
+\ell_-\mathrm{e}^{-\varrho\ell_-})
\end{equation}
with some constant $C$. Expanding now the resolvent $\mathcal{R}(\lambda,\ell)$ in relation (\ref{25}) into Taylor series, we get
\begin{equation*}
\L-\lambda_0
=\big(\mathcal{R}(\lambda_0,\ell)\Psi_0,\phi\big)_{\mathfrak{L}}
+(\lambda-\lambda_0)\big(\mathcal{R}^{'}(\lambda_0,\ell)\Psi_0,\phi\big)_{\mathfrak{L}}
+\Odr\Big(\|\mathcal{R}(\widetilde{\lambda},\ell)\|\|\Psi_0\|_{\mathfrak{L}}\Big)\,,
\end{equation*}
where $\widetilde{\lambda}$ is a point of the segment with the endpoints $\lambda_0,\,\L$ in the complex plane.
Taking into account the identity $\mathcal{R}'(\l_0,\ell)=\mathcal{R}^2(\l_0,\ell)$ --- cf.~formula (5.8) in \cite[Sec. I.5.2]{K} --- in combination with (\ref{13}), (\ref{3.23}) and Lemma~\ref{lm4}, we infer from the above relation that
\begin{equation}\label{eq4.2}
\L-\lambda_0
=\big(\mathcal{R}(\lambda_0)\Psi_0,\phi\big)_{\mathfrak{L}}
+\Odr\Big(\ell_+^2\mathrm{e}^{-3\sqrt{1-\l_0}\ell_+}
+\ell_-^2\mathrm{e}^{-3\sqrt{1-\l_0}\ell_-}\Big)\,.
\end{equation}
Next we evaluate the first term on the right-hand side, namely
\begin{equation}\label{eq4.3}
\begin{aligned}
\big(\mathcal{R}(\lambda_0)\Psi_0,\phi\big)_{\mathfrak{L}}
&=-\Big(\mathcal{T}_+\mathcal{S}(-\ell_+)\Big(\psi_0\frac{d^2 \chi_0^+}{d x_1^2}
+2\frac{\p \psi_0}{\p x_1}\frac{d \chi_0^+}{d x_1}\Big),\psi_0\Big)_{L_2(\Pi)}
\\&-\Big(\mathcal{T}_-\mathcal{S}(\ell_-)\Big(\psi_0\frac{d^2 \chi_0^-}{d x_1^2}
+2\frac{\p \psi_0}{\p x_1}\frac{d \chi_0^-}{d x_1}\Big),\psi_0\Big)_{L_2(\Pi)}\,.
\end{aligned}
\end{equation}
Using the definitions of the shift operator and the mollifiers together with the asymptotic behavior of the eigenfunctions, we get the relations
\begin{equation*}
\begin{aligned}
\mathrm{M}_\pm:=&-\int\limits_{\Pi}\mathcal{T}_\pm\mathcal{S}(\mp\ell_\pm)
\Big(\psi_0\frac{d^2 \chi_0^\pm}{d x_1^2}
+2\frac{\p \psi_0}{\p x_1}\frac{d \chi_0^\pm}{d x_1}\Big)\psi_0\di x
\\
=&-\mathrm{e}^{-\sqrt{1-\lambda_0}\,\ell_\pm}
\Psi_{0,1}^\pm
\int\limits_{\Pi}\psi_0\mathcal{T}_\pm\Big(\mathrm{e}^{\mp\sqrt{1-\lambda_0}\,x_1}\sin
x_2 \frac{d^2 \widetilde{\chi}_0^\pm}{d x_1^2}
\\&
\mp2\sqrt{1-\lambda_0}
\mathrm{e}^{\mp\sqrt{1-\lambda_0}\,x_1}\sin x_2\frac{d
\widetilde{\chi}_0^\pm}{d x_1} \Big)
+\Odr\big(\mathrm{e}^{-\sqrt{4-\lambda_0}\,\ell_\pm}\big)\,.
\end{aligned}
\end{equation*}
With the definition of $\mathcal{T}_\pm$ in view, the last relation takes the form
\begin{equation}\label{15}
\begin{aligned}
\mathrm{M}_\pm&=
-\Psi_{0,1}^\pm\mathrm{e}^{-\sqrt{1-\lambda_0}\,\ell_\pm}
\int\limits_{\Pi} \psi_0\mathcal{T}_\pm g_\pm\,\di x
+\Odr\big(\mathrm{e}^{-\sqrt{4-\lambda_0}\,\ell_\pm}\big)
\\&
=-\mathrm{e}^{-\sqrt{1-\lambda_0}\,\ell_\pm}
\Psi_{0,1}^\pm\int\limits_{\Pi}
\psi_0\Big(U_\pm\frac{d^2 \chi_\pm}{d x_1^2}
+2\frac{\p U_\pm}{\p x_1}\frac{d \chi_\pm}{d
x_1}\Big)\di x
+\Odr\big(\mathrm{e}^{-\sqrt{4-\lambda_0}\,\ell_\pm}\big)
\\&
=\mathrm{e}^{-\sqrt{1-\lambda_0}\,\ell_\pm}
\Psi_{0,1}^\pm
\int\limits_{\Pi}\big(\mathrm{A}_1^\pm-\mathrm{A}_2^\pm\big)\di x
+\Odr\big(\mathrm{e}^{-\sqrt{4-\lambda_0}\,\ell_\pm}\big)\,,
\end{aligned}
\end{equation}
where
\begin{equation*}
\begin{gathered}
g_\pm=\mathrm{e}^{\mp\sqrt{1-\lambda_0}\,x_1}\sin x_2
\frac{d^2 \widetilde{\chi}_0^\pm}{d x_1^2}
\mp2\sqrt{1-\lambda_0}
\mathrm{e}^{\mp\sqrt{1-\lambda_0}\,x_1}\sin x_2\frac{d \widetilde{\chi}_0^\pm}{d x_1}\,,
\\[.2em]
\mathrm{A}_1^\pm:=\psi_0(-\Delta-\lambda)(\chi_\pm U_\pm)\,,
\quad U_+=S(\ell_+)v_+\,,
\\[.5em]
\mathrm{A}_2^\pm:=\psi_0\chi_\pm(-\Delta-\lambda)U_\pm\,,
\quad U_-=S(-\ell_-)v_-\,,
\end{gathered}
\end{equation*}
and $v_+(x)$, $v_-(-x_1,x_2)$ are solutions of the problem (\ref{222}), (\ref{2222}) with the right-hand sides $g_+(x)$, $g_-(-x_1,x_2)$, respectively. Integrating now the expressions for $\mathrm{A}_1^\pm$ over the rectangle
$\Pi(-R,R)$ with a fixed $R$, using integration by parts and taking into account the asymptotic behavior of the functions $U_\pm$ as $r\to 0$, we get 
\begin{align*}
\int\limits_{\Pi(-R,R)} \mathrm{A}_1^\pm\di x =&-\int\limits_0^{\pi}
\psi_0(-R,x_2)\frac{\p}{\p
x_1}\chi_\pm U_\pm\Big|_{x_1=-R}\di x_2
\\
&-\int\limits_{-R}^{R}
\chi_\pm(x_1)U_\pm(x_1,0)\frac{\p \psi_0}{\p
x_2}\Big|_{x_2=0}\di x_1
\\& 
-\int\limits_0^{\pi}\psi_0(R,x_2)\frac{\p}{\p
x_2}\chi_\pm U_\pm\Big|_{x_1=R}\di x_2
\\
&+\int\limits_{-R}^{R}\psi_0(x_1,0)\frac{\p}{\p
x_2}\chi_\pm U_\pm\Big|_{x_2=0}\di x_1
\\& 
+\int\limits_0^{\pi}
\chi_\pm(-R)U_\pm(-R,x_2)\frac{\p \psi_0}{\p
x_1}\Big|_{x_1=-R}\di x_2
\\
&+
\int\limits_{\Pi(-R,R)}\chi_\pm U_\pm(-\Delta-\lambda)\psi_0\di x\,.
\end{align*}
Now we evaluate the limit of the last expression as $R\to +\infty$, take into account the equation satisfied by the eigenfunction  $\psi_0$, boundary conditions imposed on functions $\psi_0$, $U_\pm$, the definition of the cut-off function $\chi_+$, and the asymptotic behavior of the eigenfunction $\psi_0$; this yields the relations
\begin{align}\label{14}
\int\limits_{\Pi}\mathrm{A}_1^+\di x=\int\limits_{a}^{+\infty}U_+(x_1,0)\frac{\p
\psi_0}{\p x_2}\Big|_{x_2=0}\di x_1\,,
\quad
\int\limits_{\Pi}\mathrm{A}_1^-\di x=\int\limits_{-\infty}^{-a}U_-(x_1,0)\frac{\p
\psi_0}{\p x_2}\Big|_{x_2=0}\di x_1\,.
\end{align}
In the same way we find
\begin{align*}
&-\int\limits_{\Pi}\mathrm{A}_2^+\di x=
\int\limits_{-a}^{a}\psi_0(x_1,0)\frac{\p
U_+}{\p x_2}\Big|_{x_2=0}\di x_1-
\int\limits_{a}^{+\infty}U_+(x_1,0)\frac{\p
\psi_0}{\p
x_2}\Big|_{x_2=0}\di x_2\,,
\\&
-\int\limits_{\Pi}\mathrm{A}_2^-\di x=
\int\limits_{-a}^{a}\psi_0(x_1,0)\frac{\p
U_-}{\p x_2}\Big|_{x_2=0}\di x_1-
\int\limits_{-\infty}^{-a}U_-(x_1,0)\frac{\p
\psi_0}{\p
x_2}\Big|_{x_2=0}\di x_2\,.
\end{align*}
Substituting now the obtained expressions into (\ref{15}), having in mind the definitions of the functions $U_\pm$ and the asymptotic behavior of $v_\pm$, we arrive at the relation
\begin{equation}\label{17}
\begin{aligned}
\mathrm{M}_\pm&=
\mathrm{e}^{-\sqrt{1-\lambda_0}\,\ell_\pm}
\Psi_{0,1}^\pm
\int\limits_{-a}^{a}\psi_0(x_1,0)\frac{\p U_\pm}{\p
x_2}\Big|_{x_2=0}\di x_1
+\Odr\big(\mathrm{e}^{-\sqrt{4-\lambda_0}\,\ell_\pm}\big)
\\&
=\mathrm{e}^{-\sqrt{1-\lambda_0}\,\ell_\pm}
\Psi_{0,1}^\pm\int\limits_{-a}^{a}
\psi_0(x_1,0)\frac{\p }{\p
x_2}\mathcal{S}(\pm\ell_\pm)v_\pm\Big|_{x_2=0}\di x_1
+\Odr\big(\mathrm{e}^{-\sqrt{4-\lambda_0}\,\ell_\pm}\big)
\\&
=\widehat{C}_-(g_\pm)\,\mathrm{e}^{-2\sqrt{1-\lambda_0}\,\ell_\pm}
\Psi_{0,1}^\pm\int\limits_{-a}^{a}
\psi_0(x_1,0)\mathrm{e}^{\pm\sqrt{1-\lambda_0}\,x_1}\di x_1
+\Odr\big(\mathrm{e}^{-\sqrt{4-\lambda_0}\,\ell_\pm}\big)\,,
\end{aligned}
\end{equation}
where the constants $\widehat{C}_-(g_+)$ are determined by the asymptotics (\ref{2222}) of the functions $v_\pm$ as $x_1\to+\infty$. To evaluate the integral on the right-hand side of the last expression, we multiply the equation for the eigenfunction $\psi_0$ by $\mathrm{e}^{\pm\sqrt{1-\l_0}\,x_1}\sin x_2$ and integrate by parts over the region $\Pi(-R,R)$, obtaining
\begin{equation*}
\begin{aligned}
0&
=\int\limits_{\Pi(-R,R)}\mathrm{e}^{\pm\sqrt{1-\lambda_0}\,x_1}\sin
x_2 (-\Delta-\lambda_0)\psi_0\di x
=\int\limits_0^{\pi}\mathrm{e}^{\mp\sqrt{1-\lambda_0}\,R}\sin
x_2 \frac{\p \psi_0}{\p x_1}\Big|_{x_1=-R}\di x_2
\\&
-\int\limits_0^{\pi}\mathrm{e}^{\pm\sqrt{1-\lambda_0}\,R}\sin
x_2 \frac{\p \psi_0}{\p x_1}\Big|_{x_1=R}\di x_2
\mp\sqrt{1-\lambda_0}\int\limits_0^{\pi}\psi_0(-R,x_2)\,
\mathrm{e}^{\mp\sqrt{1-\lambda_0}\,R}\sin x_2\di x_2
\\&
-\int\limits_{-R}^{R}\psi_0(x_1,0)\,
\mathrm{e}^{\pm\sqrt{1-\lambda_0}\,x_1}\di x_1
\pm\sqrt{1-\lambda_0}\int\limits_0^{\pi}\psi_0(R,x_2)\,
\mathrm{e}^{\pm\sqrt{1-\lambda_0}\,R}\sin x_2\di x_2\,.
\end{aligned}
\end{equation*}
Evaluating now the limit of the last expression as $R\to+\infty$ and taking into account the asymptotic behavior of the eigenfunction, we find
\begin{equation*}
\int\limits_{-a}^{a}\psi_0(x_1,0)\,
\mathrm{e}^{\pm\sqrt{1-\lambda_0}x_1}\di x_1=-
\pi\sqrt{1-\lambda_0}\,\Psi_{0,1}^\pm\,.
\end{equation*}
In view of the even/odd character of the eigenfunction $\psi_0$ and the obvious identity $|\Psi_0|=|\Psi_{0,1}^\pm|$, the relation (\ref{17}) acquires the form
\begin{equation*}
\mathrm{M}_\pm=-\pi
\widehat{C}_-(g_\pm)\sqrt{1-\lambda_0}\big|\Psi_{0,1}^\pm\big|^2
\mathrm{e}^{-2\sqrt{1-\lambda_0}\,\ell_\pm}
+\Odr\big(\mathrm{e}^{-\sqrt{4-\lambda_0}\,\ell_\pm}\big)\,.
\end{equation*}
It remains to determine the complex constants $\widehat{C}_-(g_\pm)$. To this aim we represent the right-hand sides of the equations satisfied by $v_\pm$ as
\begin{equation*}
g_\pm=\big(-\Delta-\lambda_0\big)\big(-\widetilde{\chi}_0^\pm
\mathrm{e}^{\mp\sqrt{1-\lambda_0}x_1}
\sin x_2\big)\,.
\end{equation*}
From here it is easy to see that functions $v_\pm$ can be written as
\begin{equation*}
v_\pm(x)=-\widetilde{\chi}_0^\pm\,\mathrm{e}^{\mp\sqrt{1-\lambda_0}\,x_1}
\sin x_2+V_\pm(x)\,,
\end{equation*}
where $V_+(x)=V_j(x)$, $\:V_-(x)=V_j(-x_1,x_2)$, and $V_j(x)$ are solutions to the problem (\ref{49})--(\ref{38}). The obtained representation of $v_\pm$ implies
\begin{equation*}
\widehat{C}_-(g_\pm)=k_j^-\,,
\end{equation*}
where the complex constant $k_j^-$ is nonzero as we shall demonstrate below.

In this way we have proved the relations
\begin{align*}
\Bigg(\mathcal{T}_\pm\mathcal{S}(\mp\ell_\pm)\Big(\psi_0\frac{d^2
\chi_0^\pm}{d x_1^2} +2\frac{\p \psi_0}{\p
x_1}\frac{d \chi_0^\pm}{d
x_1}\Big),\psi_0\Bigg)_{L_2(\Pi)} &=-\pi
k_j^-\sqrt{1-\lambda_0}\big|\Psi_0\big|^2
\mathrm{e}^{-2\sqrt{1-\lambda_0}\,\ell_\pm}
\\&
+\Odr\big(\mathrm{e}^{-\sqrt{4-\lambda_0}\,\ell_\pm}\big)\,;
\end{align*}
substituting it into (\ref{eq4.3}) and taking (\ref{eq4.2}) into account, we get
\begin{align*}
\L=&\:\lambda_0-k_j^-\pi \sqrt{1-\lambda_0}|\Psi_0|^2
\Big( \mathrm{e}^{-2\sqrt{1-\lambda_0}\,\ell_+}+
\mathrm{e}^{-2\sqrt{1-\lambda_0}\,\ell_-}\Big)
\\[.3em] &
+\Odr\Big(\ell_+^2\mathrm{e}^{-2
\sqrt{1-\lambda_0}\,\ell_+} +\ell_-^2\mathrm{e}^{-2
\sqrt{1-\lambda_0}\,\ell_-}
\Big).
\end{align*}
Next we are going to prove relations (\ref{52}) and to check that $\beta_j\not=0$. To this aim we multiply equation (\ref{49})
by $\overline{V}_\pm$ and integrate it twice by parts over the region $\Pi(-R,R)\backslash w_{\varepsilon}$,
where $w_{\varepsilon}$ is the semicircle of radius $\varepsilon$ centered at the origin of coordinates. This yields
\begin{align*}
0&=\int\limits_{\Pi(-R,R)\backslash w_{\varepsilon}}\overline{V}_\pm
\big(\Delta+\lambda_0\big)\frac{\p}{\p x_1}V_\pm\di x=
-\int\limits_0^{\pi}\overline{V}_\pm\frac{\p^2}{\p x_1^2}V_\pm\Big|_{x_1=-R}\di x_2
\\&
+\int\limits_0^{\pi}\overline{V}_\pm\frac{\p^2}{\p x_1^2}V_\pm\Big|_{x_1=R}\di x_2
-\int\limits_0^{\pi}\frac{\p}{\p x_1}\overline{V}_\pm
\frac{\p}{\p x_1}V_\pm\Big|_{x_1=-R}\di x_2
+\int\limits_0^{\pi}\frac{\p}{\p x_1}\overline{V}_\pm
\frac{\p}{\p x_1}V_\pm\Big|_{x_1=R}\di x_2
\\&
-\int\limits_{\p w_{\varepsilon}}\overline{V}_\pm\frac{\p^2}{\p r \p x_1}V_\pm\di s
+\int\limits_{\p w_{\varepsilon}}\frac{\p}{\p r}V_\pm\frac{\p}{\p x_1}\overline{V}_\pm \di s\,.
\end{align*}
Evaluating each of the obtained integrals and passing to the limits $R\to+\infty$ and $\varepsilon\to 0$,
we find
\begin{align}\label{51}
\RE
k_j^-=\frac{1}{2}\big|k_j^+\big|^2\frac{\lambda_0-\frac{1}{4}}{1-\lambda_0}
-\frac{\big|\beta_j\big|^2}{8(\lambda_0-1)}\,.
\end{align}
The real and imaginary parts of the functions $V_\pm$ obviously solve the problem (\ref{49}). Their behavior as  $x_1\to\pm\infty$ and $r\to 0$ is obtained easily from the asymptotics (\ref{36}), (\ref{38}) taking their real and imaginary parts, respectively. In analogy with the above reasoning, integrating twice by parts in the identities
\begin{equation*}
\begin{gathered}
0=\int\limits_{\Pi(-R,R)\backslash w_{\varepsilon}}
\IM V_\pm
\big(\Delta+\lambda_0\big)\frac{\p}{\p x_1}\IM V_\pm\di x\,,
\\
0=\int\limits_{\Pi(-R,R)\backslash w_{\varepsilon}}
\RE V_\pm
\big(\Delta+\lambda_0\big)\frac{\p}{\p x_1}\IM V_\pm\di x\,,
\end{gathered}
\end{equation*}
we arrive at the following formul{\ae}
\begin{align}
&
\big|k_j^+\big|^2=\frac{\IM^2\beta_j}{4\big(\lambda_0-\frac{1}{4}\big)}\,,
\label{56}
\\
&
\IM
k_j^-=\frac{1}{2}\big|k_j^+\big|^2\sqrt{\frac{\lambda_0-\frac{1}{4}}{1-\lambda_0}}\,;
\label{55}
\end{align}
substituting now (\ref{56}) into (\ref{55}) and (\ref{51}) we get (\ref{52}). It remains to prove that the constants $\IM
k_j^-$ and $\Psi_{j}$ do not vanish. This is the contents of the following lemmata.
\begin{lemma}\label{lm7}
The coefficient $\Psi_{j}$ in the asymptotics (\ref{31}) is nonzero.
\end{lemma}
\begin{proof}
We use \emph{reductio ad absurdum}; we shall suppose that $\Psi_{j}$ vanishes. We multiply the equation in (\ref{49}) by $\mathrm{e}^{\pm\sqrt{1-\lambda_{0}}\,x_{1}}\sin x_{2}$ and integrate twice by parts over $\Pi(-R,R)$. Then we take into account the asymptotics (\ref{31}) with $\Psi_j=0$ and pass to the limit $R\to\infty$ obtaining
\begin{equation}\label{21}
\int\limits_{-a}^{a}\mathrm{e}^{\pm\sqrt{1-\lambda_{0}}\,x_{1}} \psi_{j}(x_{1},0)\,\di x_{1}=0\,.
\end{equation}
Let the function $\psi_{j}$ be even. We define one more function, namely
\begin{equation*}
\Phi_{j}(x):=\sinh\sqrt{1-\lambda_{0}}\,x_{1}
\int\limits_{0}^{x_{1}}\psi_{j}(t,x_{2})\cosh\sqrt{1-\lambda_{0}}\,t\,\di t\,.
\end{equation*}
By (\ref{21}) the function $\Phi_{j}$ satisfies the boundary conditions
\begin{equation*}
\Phi_{j}=0 \quad \text{on} \;\; \gamma\cup\gamma_{a}\,,\quad
\frac{\p \Phi_{j}}{\p x_{2}}=0 \quad \text{on} \;\;
\Gamma_{a}\,,
\end{equation*}
and the equation which is easy to derive from (\ref{49}),
\begin{equation}\label{4.11a}
(\Delta+\lambda_{0})\Phi_{j}=2\sqrt{1-\lambda_{0}}\cosh\sqrt{1-\lambda_{0}}\,x_{1}\,\psi_{j}
\quad\text{in}\;\;\Pi\,.
\end{equation}
It is also clear that $\Phi_j\in\H^2(\Pi(-R,R))$ holds for any $R>0$.
From (\ref{31}) with $\Psi_j=0$ and the definition of the function $\Phi_j$ it follows that
\begin{equation*}
\Phi_j(x)=\Odr(\E^{\sqrt{1-\l_0}|x_1|})\,,\quad \frac{\Phi_j}{\p x_1}(x)=\Odr(\E^{\sqrt{1-\l_0}|x_1|})\quad \text{as}\;\;|x_1|\to\infty\,.
\end{equation*}
Now we multiply equation (\ref{4.11a}) by $\psi_j$, integrate twice by parts over $\Pi(-R,R)$ and pass to the limit $R\to+\infty$ taking into account the last asymptotics and (\ref{31}) with $\Psi_j=0$. This yields
\begin{equation*}
0
=2\sqrt{1-\lambda_{0}}\int\limits_{\Pi}\psi_j^2(x)
\cosh\sqrt{1-\lambda_{0}}\,x_{1}\,\di x\geqslant 0\,,
\end{equation*}
implying that $\psi_j(x)\equiv0$ which is impossible.

In case of an odd $\psi_j$ we proceed in the analogous way replacing the auxiliary function $\Phi_j$ by
\begin{equation*}
\Phi_{j}(x):=\cosh\sqrt{1-\lambda_{0}}x_{1}
\int\limits_{0}^{x_{1}}t\psi_{j}(t,x_{2})\sinh\sqrt{1-\lambda_{0}}\,t\,\di t\,,
\end{equation*}
arriving again at the absurd conclusion $\psi_j(x)\equiv0$.
\end{proof}

\begin{lemma}\label{lm8}
The constant $k_{j}^{+}$ in relation (\ref{36}) is nonzero.
\end{lemma}
\begin{proof}
We will again proceed by contradiction. We suppose that $k_{j}^{+}$ vanishes and consider the function $\chi$ satisfying the equations
\begin{equation*}
\chi(x_{1})=1 \quad \text{for} \;\; x_{1}<0\,, \qquad
\chi(x_{1})=0 \quad \text{for} \quad x_{1}>0\,.
\end{equation*}
In view of the standard embedding theorems the function
\begin{equation}\label{33}
\widehat{V}_{j}(x)=V_{j}(x)-\chi(x_1)\E^{-\sqrt{1-\lambda_{0}}\,x_1}\sin x_2
\end{equation}
belongs to the spaces $\H^2(\Pi(-\infty,0))$ and $\H^2(\Pi(0,R))$ for all $R>0$, and at the point $x_1=0$ these functions and their first derivatives with respect to $x_1$ have a first-kind discontinuity. It is not difficult to check that $\widehat{V}_{j}(x)$ solve the equation (\ref{49}) in the region $\Pi(-\infty,0)\cup\Pi(0,+\infty)$ and exhibit the asymptotic behavior
\begin{equation}\label{4.12a}
\begin{aligned}
&V_j(x)=k^+_j\,\mathrm{e}^{\iu
\sqrt{\lambda_0-\frac{1}{4}}\,x_1}\cos
\frac{x_2}{2}+\Odr\big(\mathrm{e}^{- \sqrt{\frac{9}{4}-\lambda_0}\,x_1} \big)\,,
  &&   x_1\to+\infty\,,
\\[.2em]
&
V_j(x)=k^-_j\,\mathrm{e}^{\sqrt{1-\lambda_0}\,x_1}\sin x_2
+\Odr\big(\mathrm{e}^{- \sqrt{4-\lambda_0}\,x_1} \big)\,,
&&  x_1\to-\infty\,,
\end{aligned}
\end{equation}
where $k_j^+=0$.

Now we multiply the equation in (\ref{49}) by $\mathrm{e}^{-\iu\sqrt{\lambda_{0}-\frac{1}{4}}\,x_1}\cos\frac{x_{2}}{2}$, integrate twice by parts over $\Pi(-R,R)$ and take into account the boundary condition imposed on $\widehat{V}_{j}$ and the asymptotic behavior of these functions. Passing to the limit $R\to\infty$ we get
\begin{equation}\label{20}
\int\limits_{-\infty}^{0}\mathrm{e}^{-\iu\sqrt{\lambda_{0}-\frac{1}{4}}\,x_{1}}
\frac{\p \widehat{V}_{j}}{\p
x_{2}}\Big|_{x_{2}=0}\,\di x_{1}=\frac{4}{3}\left(\sqrt{1-\lambda_{0}}-
\iu\sqrt{\lambda_{0}-\frac{1}{4}}\right).
\end{equation}
It is easy to check that
\begin{equation}\label{35}
\int\limits_{-\infty}^{0}\mathrm{e}^{\left(-\iu\sqrt{\lambda_{0}-\frac{1}{4}}
+\sqrt{1-\lambda_{0}}\right)\,x_{1}}\,\di x_{1}
=\frac{4}{3}\Big(\sqrt{1-\lambda_{0}}+\iu\sqrt{\lambda_{0}-\frac{1}{4}}\Big).
\end{equation}
Next we define the function
\begin{equation*}
\widetilde{V}_{j}(x)=\widehat{V}_{j}(x)-C_{1}\chi(x_1)\, \mathrm{e}^{\sqrt{1-\lambda_{0}}\,x_{1}}\sin x_{2}\,, \quad
C_1:=\frac{\sqrt{1-\lambda_{0}}-
\iu\sqrt{\lambda_{0}-\frac{1}{4}}}
{\sqrt{1-\lambda_{0}}+\iu\sqrt{\lambda_{0}-\frac{1}{4}}}\,,
\end{equation*}
which also belongs to the spaces $\H^2(\Pi(-\infty,0))$ and $\H^2(\Pi(0,R))$ for any \mbox{$R>0$}. In view of (\ref{20}) the function $\widetilde{V}_{j}$ solves the equation from (\ref{49}) in the region $\Pi(-\infty,0)\cup\Pi(0,+\infty)$ and satisfies the boundary conditions from (\ref{49}). It follows from (\ref{20}) and (\ref{35}) that
\begin{equation}\label{4.14a}
\int\limits_{-\infty}^{0}\mathrm{e}^{-\iu\sqrt{\lambda_{0}-\frac{1}{4}}\,x_{1}}
\frac{\p \widetilde{V}_{j}}{\p
x_{2}}\Big|_{x_{2}=0}\,\di x_{1}=0\,.
\end{equation}
It is not difficult to check that
\begin{equation*}
\widetilde{V}_{j}(0,x_2)=
\frac{2 \sqrt{1-\lambda_{0}}}
{\sqrt{1-\lambda_{0}}-\iu\sqrt{\lambda_{0}-\frac{1}{4}}}\sin
x_{2}\,,\quad
\frac{\p \widetilde{V}_{j}}{\p x_{1}}(0,x_2)=
\frac{2\iu\sqrt{1-\lambda_{0}}\sqrt{\lambda_{0}-\frac{1}{4}}}
{\sqrt{1-\lambda_{0}}-\iu\sqrt{\lambda_{0}-\frac{1}{4}}}\sin
x_{2}\,.
\end{equation*}
Consider now the function
\begin{equation}\label{30}
Z_{j}(x):=\mathrm{e}^{\iu\sqrt{\lambda_{0}-\frac{1}{4}}\,x_{1}}
\int\limits_{-\infty}^{x_1}\mathrm{e}^{-\iu\sqrt{\lambda_{0}-\frac{1}{4}}\,t}\,
\widetilde{V}_{j}(t,x_2)\,\di t\,,
\end{equation}
which belongs to $\H^2(\Pi(-\infty,0))$, $\H^2(\Pi(0,R))$ and $\H^1(\Pi(-R,R))$ for all $R>0$. In view of (\ref{4.14a}) and boundary conditions satisfied by $\widetilde{V}_j$ it is not difficult to check that the function $Z_j$ satisfies the boundary condition from (\ref{49}); by a direct computation one checks that it satisfies also the equations
\begin{align*}
&(-\Delta-\lambda_{0})Z_{j}=0\quad\text{in}\quad\Pi(-\infty,0)\,,
\\[.3em]
&(-\Delta-\lambda_{0})Z_{j}=\mathrm{e}^{-\iu\sqrt{\lambda_{0}-\frac{1}{4}}\,x_{1}}
\Bigg(\bigg[\frac{\p \widetilde{V}_{j}}{\p
x_{1}}\bigg]_{x_{1}=0}
-\iu\sqrt{\lambda_{0}-\frac{1}{4}}\big[\widetilde{V}_{j}\big]_{x_{1}=0}\Bigg)=0
\quad\text{in}\quad\Pi(0,+\infty)\,.
\end{align*}
Furthermore, the function $Z_j$ exhibits the asymptotic behavior (\ref{4.12a}) with $k_j^\pm$ replaced by other complex constants. We notice also that
\begin{equation}\label{4.14b}
Z_j(0,x_2)=0\,,
\quad\frac{\p Z_{j}}{\p
x_{1}}(0,x_2)=\widetilde{V}_{j}(0,x_2) =\frac{2\sqrt{1-\lambda_{0}}}
{\sqrt{1-\lambda_{0}}+\iu\sqrt{\lambda_{0}-\frac{1}{4}}}\sin
x_{2}.
\end{equation}
In this way, the function $Z_j$ solves the problem with the jump (\ref{4.14b}). Such a solution is unique as it is easy to deduce from Lemma~\ref{lm6}. On the other hand, the solution can be also expressed in terms of the original function $V_j$:
\begin{equation*}
Z_j(x)=\frac{1}{\sqrt{1-\lambda_{0}}-\iu\sqrt{\lambda_{0} -\frac{1}{4}}} \left(V_j(x)-\chi(x_1)
\Big(\mathrm{e}^{-\sqrt{1-\lambda_{0}}\,x_{1}}
-\mathrm{e}^{\sqrt{1-\lambda_{0}}\,x_{1}}\Big)\sin x_2\right).
\end{equation*}
From here and definition (\ref{30}) one can derive the differential equation
\begin{equation*}
\frac{\p \phi_{j}}{\p x_{1}}+C_2\phi_{j}=
C_3\chi\mathrm{e}^{\big(\sqrt{1-\lambda_{0}}-\iu\sqrt{\lambda_{0}-\frac{1}{4}}\big)x_{1}}\sin
x_{2},
\end{equation*}
where
\begin{equation*}
\begin{gathered}
\phi_{j}=\mathrm{e}^{-\iu\sqrt{\lambda_{0}-\frac{1}{4}}\,x_{1}}\widetilde{V}_{j}\,,
\quad
C_2=\sqrt{1-\lambda_{0}}-\iu\sqrt{\lambda_{0}-\frac{1}{4}}\,,
\\[.3em]
C_3=\frac{2\sqrt{1-\lambda_{0}}\big(\sqrt{1-\lambda_{0}}
-\iu\sqrt{\lambda_{0}-\frac{1}{4}}\big)}
{\sqrt{1-\lambda_{0}}+\iu\sqrt{\lambda_{0}-\frac{1}{4}}}\,;
\end{gathered}
\end{equation*}
solving it we arrive at
\begin{equation*}
V_{j}(x)=\mathrm{e}^{-\sqrt{1-\lambda_{0}}\,x_{1}}\Big(C_4(x_{2})+\chi\sin
x_{2}\Big).
\end{equation*}
The function $V_{j}$ satisfies the boundary-value problem (\ref{49}) and the asymptotic requirement (\ref{36}), hence we necessarily have
\begin{equation*}
V_{j}=\mathrm{e}^{-\sqrt{1-\lambda_{0}}x_{1}}\chi\sin x_{2}\,,\quad \pm x_1>0\,,
\end{equation*}
however, such a function does not belong to $\Hloc^1(\Pi)$, and therefore it cannot represent a generalized solution to the problem (\ref{49}), (\ref{36}).
\end{proof}

It follows from the last lemma and identity (\ref{55}) that the imaginary part of the constant $k_{j}^{-}$ does not vanish. This concludes the proof of Theorem~\ref{th1}.

\subsection*{Acknowledgements}

The research was supported by Russian Foundation for Basic Research and by Czech Science Foundation under the contract P203/11/0701.


\end{document}